\documentclass[11pt,a4paper]{article}

\usepackage[english]{babel}
\usepackage[T1]{fontenc}
\usepackage[a4paper,margin=25mm]{geometry}
\usepackage{microtype} 
\usepackage[colorlinks=true, linkcolor=blue, citecolor=blue, urlcolor=blue]{hyperref}

\usepackage{amsmath,amssymb,amsthm}
\usepackage{mathtools}
\usepackage{tikz-cd}

\numberwithin{equation}{section}
\theoremstyle{plain}
\newtheorem{theorem}{Theorem}[section]

\newtheorem{proposition}[theorem]{Proposition}
\newtheorem{corollary}[theorem]{Corollary}
\newtheorem{conjecture}[theorem]{Conjecture}

\theoremstyle{definition}
\newtheorem{definition}[theorem]{Definition}

\DeclareMathOperator{\GL}{GL}
\DeclareMathOperator{\SL}{SL}
\DeclareMathOperator{\Mat}{Mat}
\DeclareMathOperator{\End}{End}
\DeclareMathOperator{\Hom}{Hom}
\DeclareMathOperator{\Id}{Id} 
\DeclareMathOperator{\ad}{ad}
\DeclareMathOperator{\tr}{tr}
\DeclareMathOperator{\im}{im} 
\DeclareMathOperator{\diag}{diag}
\DeclareMathOperator{\spec}{spec} 

\newcommand{\C}{\mathbb{C}}
\newcommand{\Z}{\mathbb{Z}}

\title{Isomonodromic Deformations \\
  for Linear \(q\)-Difference Systems of Degree One}
\author{Yiming Ma}
\date{}
\newcommand{\Addresses}{{
  \bigskip
  \footnotesize
  \par\noindent
  \textsc{School of Mathematical Sciences, Peking University, Beijing 100871, China}
  \par\nopagebreak
  \textit{E-mail address}: \texttt{ymma@pku.edu.cn}
}}

\begin{document}

\maketitle
\begin{abstract}
We construct isomonodromy transformations for linear \(q\)-difference systems of the form \(Y(qz)=A(z)Y(z)\), where \(A(z)=A_0+z A_1\) has diagonal leading coefficient. These transformations shift eigenvalues of the leading coefficient \(A_1\) together with roots of \(\det A(z)\). They lift compatibility to the right eigenpairs of \(A(z)\), yielding a discrete local tau function. The resulting deformation equations preserve the Birkhoff connection matrix, and reduce in their \(q\to 1\) limit to the isomonodromic deformation of a meromorphic connection on \(\mathbb P^1\) with an irregular singularity of Poincaré rank one at \(\infty\).
\end{abstract}

\section{Introduction}

Consider linear \(q\)-difference systems of the form \(Y(qz)=A(z)Y(z)\), with rational coefficient \(A(z)\in\GL_n(\C(z))\). Any such system Fuchsian at \(0\) and \(\infty\) reduces by rational gauge transformation to a polynomial form \(A(z)=A_0+z A_1+\cdots+z^\mu A_\mu\), with \(A_0,A_\mu\in\GL_n(\C)\) encoding the local exponents at \(0\) and \(\infty\) \cite{ohyama2020space,sauloy2000systemes}.
Under nonresonance conditions, Birkhoff constructed canonical fundamental solutions \(Y_0\) and \(Y_\infty\), whose connection matrix \(P=Y_0^{-1}Y_\infty\) is pseudo-constant in \(z\). He showed that with fixed local exponents, \(P(z)\) determines the system up to rational equivalence \cite{birkhoff1913generalized}.
Accordingly, a discrete transformation of \(A(z)\) preserving \(P(z)\) is called an isomonodromy transformation or a connection-preserving deformation.

We restrict ourselves to degree-one systems \(A(z)=A_0+z A_1\) of arbitrary rank. Their role is already visible in Mano's study of the connection problem for the \(q\)-Painlev\'e VI equations. Jimbo and Sakai obtained these equations as the compatibility condition for a connection-preserving deformation of a \(2\times2\) linear \(q\)-difference system of degree two \cite{jimbo1996q}.
Near a boundary point of the deformation, Mano proved that the connection matrix of this system factors into those of two \(q\)-difference systems of degree one, each solvable in terms of Heine's basic hypergeometric series \cite{manoAsymptoticBehaviourBoundary2010,ohyamaAnalyticSolutions2009}.
Ohyama, Ramis and Sauloy further interpret this factorization as localizing monodromy at pairs of intermediate singularities, namely the zeros of \(\det A(z)\), and regard its extension to higher-rank systems as a major open problem \cite{ohyama2020space}.
We expect to provide such a generalization by combining our isomonodromy construction here with the explicit \(q\)-Stokes analysis for the degenerate case \(A_1=\kappa_n E_{nn}\) in \cite{linMaXu2024explicit}.

Systems of degree one also arise from hierarchy reductions. Kakei and Kikuchi proposed a \(q\)-analogue of the \(\widehat{\mathfrak{gl}}_n\) Drinfel'd-Sokolov hierarchy as a reduction of the \(q\)-KP hierarchy \cite{kakeiKikuchi2006qAnalogue}.
Such systems are then obtained under a similarity reduction, with their deformation equations expressed in terms of a Sato-Wilson operator.
This motivates our explicit isomonodromy construction of each deformation step directly from the coefficient \(A(z)\), as well as our choice of normalization for \(Y_\infty\) which makes \(\det P(z)\) constant.
In particular, they showed that for \(n=3\), the resulting \(q\)-Lax system is equivalent to the Jimbo-Sakai family for \(q\)-Painlev\'e VI after a \(q\)-Laplace transformation.

Our construction is based on Borodin's treatment for linear difference systems \cite{borodin2004isomonodromy}, where he introduced the transformations \(A(z)\mapsto R(z+1)A(z)R(z)^{-1}\) by choosing \(R(z)\) as matrix right divisors of \(A(z)\).
While Borodin's general treatment fixes the diagonal leading coefficient, which is \(A_1\) here, possibly up to a constant scalar factor, we deform its selected eigenvalues \(\kappa_i\) together with the roots \(x_i\) of \(\det A(z)\) paired with them.

Write \(\kappa=(\kappa_1,\dots,\kappa_n)\) and \(x=(x_1,\dots,x_n)\), with \(A_1=D(\kappa)\) and \(X=D(x)\) the corresponding diagonal matrices. 
For \(I\subseteq [n]\coloneqq \{1,\dots,n\}\), set \(E_I \coloneqq \sum_{i\in I}E_{ii}\), and define \(T_I A(z) \coloneqq A_I(qz)A(z)A_I(z)^{-1}\) by a matrix right divisor \(A_I(z)\) with prescribed form and determinant
\begin{equation*}
  A_I(z)=-zX^{-1}E_I+Q_I, \qquad
  \det A_I(z)=\prod_{i\in I}\left(1-\frac{z}{x_i}\right).
\end{equation*}
Such a matrix divisor is constructed in Theorem~\ref{thm:Construction} by a reduction to block-triangular form. Then an algebraic Riccati equation determines \(A_I(z)\) up to a constant gauge, and a Sylvester equation further fixes that gauge by requiring the leading coefficient of \(T_I A(z)\) to be diagonal.
The resulting transformation sends \((x_i,\kappa_i)\) to \((q^{-1}x_i,q\kappa_i)\) for \(i\in I\) while fixing the remaining pairs, yielding local exponents \(q^{\alpha_i}=-\kappa_i x_i\) for \(i\in[n]\). It also sends \(A_0\) to \(Q_I A_0 Q_I^{-1}\), yielding local exponents \(q^{\beta_i}=\theta_i\) (\(i\in[n]\)) from the eigenvalues \(\theta_i\) of \(A_0\). 
The action of \(T_I\) natural extends to the right eigenpairs \((v_j,x_j)\) of \(A(z)\), for which \(0\neq v_j\in \ker A(x_j)\). We denote them by a ``canonical pair'' \((V,X)\) with \(V=(v_1,\dots,v_n)\in \GL_n(\C)\), then the domain of \(T_I\) is characterized by the principal minor \(\Delta_I(V)\neq 0\). Consequently, \(T_I\) is defined on a Zariski open subset of a smooth irreducible affine variety \(E_{\alpha}(x)\) and restrict to the subset of \(E_{\alpha,\beta}(x)\) (see Definition~\ref{def:coefficient-space}).

Flatness of the transformations follows from a composition law in Theorem~\ref{thm:Composition}. For disjoint \(I,J\subseteq [n]\), we have \(T_J T_I = T_{I\sqcup J}\) whenever the composition is defined. Interchanging \(I,J\) gives \(T_I T_J=T_J T_I\), acting compatibility on both \(A(z)\) and the pair \((V,X)\). In particular, given a flat canonical pair \(\bigl(V(x),D(x)\bigr)\), one can define the local tau function by \(\tau(T_I x)/\tau(x)=\Delta_I(V(x))\).

Now we consider a family of coefficients \(A(z;x)\in E_{\alpha,\beta}(x)\) parametrized by the roots of their determinants. It is well-known that the connection matrix \(P(z;x)=Y_0^{-1}Y_{\infty}(z;x)\) is invariant under a shift operator \(T_I\), if and only if the canonical solutions \(Y_0\) and \(Y_\infty\) satisfy \(T_I Y = B(z;x)Y\) for some rational \(B(z;x)\) \cite{jimbo1996q}. Under the strong nonresonance conditions, this \(B(z;x)\) must concide with \(A_I(z;x)\) constructed above by Theorem~\ref{thm:Isomonodromy}.
Consequently, our construction gives the unique connection preserving deformation sending \(A(z;x)\in E_{\alpha,\beta}(x)\) to \(A(z;T_I x)\in E_{\alpha,\beta}(T_I x)\).

Finally, by setting \(A_0=\Id_n+(q-1) A\) and \(A_1=(q-1) U\), our construction formally reduces in their \(q\to 1\) limit to the form
\begin{equation*}
  \frac{\partial Y}{\partial z} =
    \left(U+\frac{A}{z}\right)Y, \qquad
  \frac{\partial Y}{\partial u_i} =
    \left(zE_i+\ad_U^{-1}[E_i,A]\right)Y
  \quad \text{for } 1\le i\le n,
\end{equation*}
where \(E_i \coloneqq E_{ii}\), and \(\ad_U^{-1}\) denotes the inverse of \(\ad_U\) on off-diagonal matrices. This limiting system is the classical Jimbo-Miwa-Ueno isomonodromic deformation of a meromorphic connection on \(\mathbb P^1\), with a Fuchsian singularity at \(0\) and an irregular singularity of Poincaré rank one at \(\infty\) \cite{jimboMiwaUeno1981}. It also appears in the theories of semisimple Frobenius manifolds \cite{dubrovinGeometry2DTopological1996}, Poisson-Lie groups \cite{boalchStokesMatricesPoisson2001}, and Bridgeland stability conditions \cite{bridgelandStabilityConditionsStokes2012}.

Under our chosen normalization of \(Y_0(z)\) and \(Y_\infty(z)\), the connection matrix \(P(z)\) has constant determinant, putting a strong constraint on the characteristic constants \(p_{ij}\). It is expected that these characteristic constants can be determined recursively from our isomonodromy constructions. The resulting formulas are then \(q\)-analogues of the monodromy formulas in \cite{tangXuBoundaryCondition2024,xuRegularizedLimits2024}, thereby extending Mano's decomposition to higher-rank systems.

This paper is organized as follows.
In Section \ref{sec:2}, we study basic properties of the Birkhoff connection matrix, and review in Theorem~\ref{thm:RHB-correspondence} the Riemann-Hilbert-Birkhoff correspondence.
In Section~\ref{sec:3}, we give the explicit construction of isomonodromy transformations in Theorem~\ref{thm:Construction}.
In Section~\ref{sec:4}, we prove the composition law and 
its consequences.
We then show in Theorem~\ref{thm:Isomonodromy} that the resulting deformations preserve the Birkhoff connection matrix.
In Section~\ref{sec:5}, we derive the classical isomonodromy equations as a \(q\rightarrow 1\) formal continuous limit.
An explicit rank-two example is presented in the appendix.

\section{Linear \texorpdfstring{\(q\)}{q}-Difference Systems}
\label{sec:2}

Let \(q\in\C^*\) with \(0<|q|<1\), and fix a branch of \(\log q\). Write \(\sigma_qf(z) \coloneqq f(qz)\). Define the \(q\)-Pochhammer symbols and the \(q\)-theta function by
\begin{equation*}
  (u;q)_\infty \coloneqq \prod_{m=0}^\infty(1-q^m u),
  \quad \text{for } u\in\C, \qquad
  \theta_q(u) \coloneqq  \sum_{m\in\Z} q^{m(m-1)/2}u^m,
  \quad \text{for } u\in\C^*.
\end{equation*}
In particular, \(\theta_q(u)\) satisfies \(\theta_q(qu)=u^{-1}\theta_q(u)\), and by Jacobi's triple product identity
\begin{equation*}
  \theta_q(u) = (q;q)_\infty (-u;q)_\infty (-q/u;q)_\infty,
\end{equation*}
it has simple zeros along the \(q\)-spiral \(-q^{\Z}\) and nowhere else. Both of these functions extend naturally to diagonal matrices. For \(a=(a_1,\dots,a_n)\in\C^n\), define \(D(a) \coloneqq \diag(a_1,\dots,a_n)\).

\subsection{Canonical solutions}

Let \(\mathcal V=\C^n\) with standard basis \(\{e_i\}_{i=1}^n\), and identify \(\GL(\mathcal V)\) with \(\GL_n(\C)\). Consider a linear \(q\)-difference system with polynomial coefficient
\begin{equation}\label{eq:system}
  Y(qz)=A(z)Y(z), \quad
  A(z)=A_0+z A_1
  \in\Mat_n(\C[z]),
\end{equation}
with \(A_0,A_1\in\GL_n(\C)\).
To construct the canonical solutions, we assume that \(A_0\) and \(A_1\) are semisimple, and write \(\theta=(\theta_1,\dots,\theta_n)\), \(\kappa=(\kappa_1,\dots,\kappa_n)\) for their ordered tuples of eigenvalues, respectively. After a constant gauge transformation, we may assume \(A_1=D(\kappa)\) and \(A_0=C_0D(\theta)C_0^{-1}\) for some \(C_0\in\GL(\mathcal V)\). See \cite{sauloy2000systemes} for more general treatment.

Now we order the roots of \(\det A(z)\) as \(x=(x_1,\dots,x_n)\), so that \(x_j\) is naturally paired with \(\kappa_j\). We call these roots the intermediate singularities of the system. Since \(A_1=D(\kappa)\), we have
\begin{equation*}\label{eq:system-determinant}
  \det A(z) = \kappa_1\cdots\kappa_n (z-x_1)\cdots(z-x_n).
\end{equation*}
Taking \(z=0\) then gives the Fuchs relation
\((-1)^n\prod_{j=1}^n\kappa_j x_j=\prod_{j=1}^n\theta_j\),
and this allows us to choose local exponents \(\alpha,\beta\) such that
\begin{equation}\label{eq:local-exponents}
  q^{D(\alpha)}=-D(\kappa)X, \quad
  q^{D(\beta)}=D(\theta), \quad
  \tr D(\alpha)=\tr D(\beta),
\end{equation}
where \(X=D(x)\). With these choices, the system coefficients are given by
\begin{equation}\label{eq:local-framing}
  A_0=C_0D(\theta)C_0^{-1}=C_0q^{D(\beta)}C_0^{-1},
  \qquad
  A_1=D(\kappa)=-q^{D(\alpha)}X^{-1}.
\end{equation}
Fix a branch of \(\log z\), then \(z^{D(\alpha)}\) and \(z^{D(\beta)}\) are single-valued on the universal cover of \(\C^*\). In particular, by extending the \(q\)-Pochhammer symbol to diagonal matrices, we have
\begin{equation}\label{eq:formal-frame-shift}
  \sigma_q\left(\frac{z^{D(\alpha)}}
    {(zX^{-1};q)_\infty}\right)
  =q^{D(\alpha)}\bigl(\Id_{\mathcal V}-zX^{-1}\bigr)
  \frac{z^{D(\alpha)}}{(zX^{-1};q)_\infty}
  =\bigl(q^{D(\alpha)}+z A_1\bigr)
  \frac{z^{D(\alpha)}}{(zX^{-1};q)_\infty}.
\end{equation}

\begin{proposition}[Canonical solutions]\label{prop:canonical-solutions}
Let \(A_0,A_1\) be semisimple matrices written as in \eqref{eq:local-framing}. Then under the nonresonance conditions
\begin{equation}\label{eq:nonresonance}
  q^m\theta_i\ne\theta_j, \quad
  q^m\kappa_i\ne\kappa_j, \qquad
  (m\ge1,\ 1\le i,j\le n),
\end{equation}
the system \eqref{eq:system} admits unique fundamental solutions of the form
\begin{equation}\label{eq:canonical-solutions}
  Y_0(z)=C_0Y_0^*(z)z^{D(\beta)},\qquad
  Y_\infty(z)=Y_\infty^*(z)\frac{z^{D(\alpha)}}{(zX^{-1};q)_\infty},
\end{equation}
where \(Y_0^*(z)\in\GL_n(\C\{z\})\) and \(Y_\infty^*(z)\in\GL_n(\C\{z^{-1}\})\) are convergent series in \(z\) and \(z^{-1}\) respectively, with normalization
\begin{equation}\label{eq:normalization}
  Y_0^*(0)=Y_\infty^*(\infty)=\Id_{\mathcal V}.
\end{equation}
\end{proposition}

\begin{proof}
Substituting \eqref{eq:canonical-solutions} into the equation and using \eqref{eq:formal-frame-shift}, we have
\begin{align}
  Y_0^*(qz)q^{D(\beta)}
    &=C_0^{-1}\bigl(A_0+z A_1\bigr)C_0Y_0^*(z),
    \label{eq:Y_0^*}\\
  Y_\infty^*(qz)
  \bigl(q^{D(\alpha)}+z A_1\bigr)
    &=\bigl(A_0+z A_1\bigr)Y_\infty^*(z),
    \label{eq:Y_infty^*}
\end{align}
where \(q^{D(\beta)}=C_0^{-1}A_0C_0=D(\theta)\) and \(A_1=D(\kappa)\) are invertible matrices.
Expanding \(Y_0^*(z)\in \Id + z \Mat_n(\C[[z]])\) and \(Y_0^*(z)\in \Id + z^{-1}\Mat_n(\C[[z^{-1}]])\) then gives them as unique convergent series under the nonresonance condition \eqref{eq:nonresonance}, which we omit the details.
\end{proof}

\subsection{Birkhoff connection matrices}

The holomorphic germs \(Y_0^*\) and \(Y_\infty^*\) can be analytically continued as single-valued meromorphic functions on \(\C^*\) by \eqref{eq:Y_0^*} and \eqref{eq:Y_infty^*} respectively. As a result, \(Y_0(z)\) and \(Y_\infty(z)\) are meromorphic solutions on the universal cover of \(\C^*\).

\begin{definition}\label{def:birkhoff-connection-matrix}
The Birkhoff connection matrix \(P(z)\) between \(Y_\infty(z)\) and \(Y_0(z)\) is defined by
\begin{equation}\label{eq:birkhoff-connection-matrix}
  Y_\infty(z)=Y_0(z)P(z).
\end{equation}
In particular, \(P(z)\) is pseudo-constant in \(z\), that is, \(P(qz)=P(z)\).
\end{definition}

\begin{proposition}\label{prop:birkhoff-entires}
The connection matrix \(P(z)\) is of the form
\begin{equation}\label{eq:birkhoff-entries}
  P(z) = (P_{ij}(z))_{1 \le i,j \le n}, \quad
  P_{ij}(z)
  = p_{ij}
  \frac{\theta_q(\kappa_j z/q^{\beta_i})}
    {\theta_q(\kappa_j z/q^{\alpha_j})}
  z^{\alpha_j-\beta_i},
\end{equation}
where \(p_{ij}\in\C\) are called the characteristic constants. In particular, each \(P_{ij}(z)\) has at most simple poles, which are contained in the discrete \(q\)-spiral \(x_jq^{\Z}\).
\end{proposition}

\begin{proof}
Equation~\eqref{eq:Y_0^*} gives
\begin{equation*}
  Y_0^*(z)^{-1}
  = q^{-D(\beta)}Y_0^*(qz)^{-1}C_0^{-1}A(z)C_0,
\end{equation*}
so the holomorphic germ \(Y_0^*(z)^{-1}\) extends holomorphically to \(\C^*\). Similarly, equation~\eqref{eq:Y_infty^*} gives
\begin{equation*}
  Y_\infty^*(qz)
  = A(z)Y_\infty^*(z)q^{-D(\alpha)}
  \bigl(\Id_{\mathcal V}-zX^{-1}\bigr)^{-1},
\end{equation*}
so the \(j\)-th column of \(Y_\infty^*(z)\) extends meromorphically to \(\C^*\), with at most simple poles contained in \(\{x_jq^m:m\ge 1\}\).
Since \(1/(z x_j^{-1};q)_\infty\) has only simple poles along \(z\in\{x_jq^m:m\leq 0\}\), by equation~\eqref{eq:canonical-solutions} the \(j\)-th column of
\begin{equation*}
  z^{D(\beta)}P(z)z^{-D(\alpha)} =
  Y_0^*(z)^{-1}C_0^{-1}Y_\infty^*(z) 
  \frac{1}{(zX^{-1};q)_\infty}
\end{equation*}
is a single-valued meromorphic function on \(\C^*\) with at most simple poles contained in \(x_jq^{\Z}\).
Now by \(q^{\alpha_j}=-\kappa_jx_j\), \(\theta_q(\kappa_jz/q^{\alpha_j})=\theta_q(-z/x_j)\) vanishes precisely on \(x_jq^{\Z}\). Therefore for each \(i,j\), there exists holomorphic functions \(h_{ij}\in\mathcal O(\C^*)\) such that
\begin{equation}\label{eq:birkhoff-holomorphic}
  z^{\beta_i}P_{ij}(z)z^{-\alpha_j}
    =\frac{h_{ij}(z)}{\theta_q(\kappa_jz/q^{\alpha_j})},
  \quad
  h_{ij}(qz)
    =\frac{q^{\beta_i}}{\kappa_jz}h_{ij}(z),
\end{equation}
where the second identity follows from the pseudo-constancy of \(P(z)\). Inserting the Laurent expansion of \(h_{ij}\) then gives \(h_{ij}(z)=p_{ij}\theta_q(\kappa_jz/q^{\beta_i})\) for some constant \(p_{ij}\in\C\).
\end{proof}

\begin{corollary}\label{cor:constant-connection-determinant}
The Birkhoff connection matrix \(P(z)\) has constant determinant, given by
\begin{equation}\label{eq:birkhoff-constant-determinant}
  \det P(z)=(\det C_0)^{-1}.
\end{equation}
\end{corollary}

\begin{proof}
Taking determinants in \eqref{eq:Y_0^*} and \eqref{eq:Y_infty^*}, one can verify that
\begin{equation*}
  \det Y_0^*(z)\prod_{j=1}^n(z/x_j;q)_\infty
  \quad \text{and} \quad
  \det Y_\infty^*(z)
\end{equation*}
are pseudo-constant in \(z\), hence equal to \(1\) by taking their \(z\rightarrow 0\) and \(\infty\) limits respectively according to the normalization \eqref{eq:normalization}. Thus we have
\begin{equation}\label{eq:det Y(z)}
  \det Y_0(z) = \det C_0 \frac{z^{\tr D(\beta)}}
    {\prod_{j=1}^n(z/x_j;q)_\infty},
  \qquad
  \det Y_\infty(z) = \frac{z^{\tr D(\alpha)}}
    {\prod_{j=1}^n(z/x_j;q)_\infty},
\end{equation}
whose ratio gives \eqref{eq:birkhoff-constant-determinant} since \(\tr D(\alpha)=\tr D(\beta)\).
\end{proof}

\subsection{The Riemann-Hilbert-Birkhoff correspondence}

Proposition~\ref{prop:birkhoff-entires} writes \(P(z)\) in terms of the exponents \(\alpha,\beta\) and the ordered intermediate singularities \(x=(x_1,\dots,x_n)\).
To parametrize the space of system coefficients with these prescribed roots and exponents, we introduce the following notion of a canonical pair.

\begin{definition}\label{def:canonical-pair}
Let \(A(z)=A_0+z A_1\), with \(A_0,A_1\in\GL(\mathcal V)\), and let \(x=(x_1,\dots,x_n)\) be the ordered roots of \(\det A(z)\), counted with multiplicity.
A canonical pair of \(A(z)\) associated with \(x\) is an eigendecomposition \((V,X)\) of \(-A_1^{-1}A_0\), which is specified by the condition
\begin{equation}\label{eq:canonical-pair}
  V=(v_1,\dots,v_n)\in\GL(\mathcal V), \qquad
  X=D(x), \qquad
  v_j\in\ker A(x_j) \quad (1\le j\le n).
\end{equation}
We call \(V\) the corresponding canonical frame.
In particular, the space of coefficients \(A(z)\) admitting such a canonical pair is given by
\begin{equation*}\label{eq:E(x)}
  E(x)  \coloneqq
  \left\{
    A(z)=A_1\bigl(z\Id_{\mathcal V}-VXV^{-1}\bigr):
    A_1,V\in\GL(\mathcal V)
  \right\}.
\end{equation*}
\end{definition}
Here for fixed \(A(z)\) and \(x\), \(V\) is uniquely defined up to right multiplication by an invertible matrix commuting with \(X\).
Now we restrict to the following subspaces of \(E(x)\), with diagonal leading coefficient \(A_1\) and prescribed local exponents.

\begin{definition}\label{def:coefficient-space}
For the exponents \(\alpha,\beta\) and ordered roots \(x\) as in \eqref{eq:local-exponents}, define the following spaces of system coefficients:
\begin{equation*}
  \begin{aligned}
    E_\alpha(x)
    & \coloneqq  \left\{
      A(z)\in E(x): 
      A_1=-q^{D(\alpha)}X^{-1} \right\},\\
    E_{\alpha,\beta}(x)
    & \coloneqq \left\{
      A(z)\in E_\alpha(x): 
      A_0=C_0q^{D(\beta)}C_0^{-1}
      \text{ for some }C_0\in\GL(\mathcal V)
    \right\}.
  \end{aligned}
\end{equation*}
\end{definition}

In particular, if \((V,X)\) is a canonical pair of \(A(z)\in E_\alpha(x)\), then \(A(z)\) is parametrized by
\begin{equation}\label{eq:A(z)-parametrization}
  A(z) = q^{D(\alpha)}X^{-1}
  \bigl(VXV^{-1}-z\Id_{\mathcal V}\bigr).
\end{equation}
The map \(A(z)\mapsto-A_1^{-1}A_0\) identifies \(E_\alpha(x)\) with the \(\GL(\mathcal V)\)-conjugacy class of \(X\).
Since \(X\) is semisimple, this conjugacy class and hence \(E_\alpha(x)\) is a smooth irreducible affine variety.

As is shown in Proposition~\ref{prop:birkhoff-entires} and Corollary~\ref{cor:constant-connection-determinant}, the connection matrix of a nonresonant system in \(E_{\alpha,\beta}(x)\) is of the form \eqref{eq:birkhoff-entries} and has constant determinant. Although the \(p_{ij}\) are constants, the condition \(\det P(z)\equiv c_P\) imposes highly nontrivial relations among them.
This leads to the following space of monodromy data.

\begin{definition}\label{def:monodromy-data-space}
For the exponents \(\alpha,\beta\), ordered roots \(x\), and \(\kappa_j=-q^{\alpha_j}/x_j\) as in \eqref{eq:local-exponents}, define the space of monodromy data
\begin{equation*}\label{eq:F(x)}
  F_{\alpha,\beta}(x)
   \coloneqq
  \left\{
    P(z) = \left(p_{ij}\frac{\theta_q(\kappa_j z/q^{\beta_i})}
      {\theta_q(\kappa_j z/q^{\alpha_j})}z^{\alpha_j-\beta_i}
    \right)_{1\le i,j\le n}
    \;\middle|\;
    \begin{gathered}
      p_{ij}\in\C,\\
      \det P(z)\equiv c_P \in\C^*
    \end{gathered}
  \right\}.
\end{equation*}
\end{definition}

The following theorem is the Riemann-Hilbert-Birkhoff correspondence due to Ohyama, Ramis and Sauloy \cite[Theorem~3.7]{ohyama2020space}, written in our normalization of \(P(z)\).

\begin{theorem}[Riemann-Hilbert-Birkhoff correspondence]
\label{thm:RHB-correspondence}
Under the strong nonresonance conditions
\begin{equation}\label{eq:strong-nonresonance}
  \frac{\theta_i}{\theta_j},
  \frac{\kappa_i}{\kappa_j},
  \frac{x_i}{x_j}
  \notin q^{\Z},
  \qquad \text{for } i\neq j,
\end{equation}
the Birkhoff connection matrix induces a natural bijection
\begin{equation}\label{eq:RHB-correspondence}
  E_{\alpha,\beta}(x)/\sim_{\mathrm{rat}}
  \;\xrightarrow{\;\sim\;}\;
  F_{\alpha,\beta}(x)/\sim_{\mathrm{diag}}.
\end{equation}
Here \(A\sim_{\mathrm{rat}}A'\) means \(A'(z)=G(qz)A(z)G(z)^{-1}\) for some \(G(z)\in\GL_n(\C(z))\), while \(P\sim_{\mathrm{diag}}P'\) means \(P'(z)=\Gamma^{-1}P(z)\Delta\) for some invertible diagonal matrices \(\Gamma\) and \(\Delta\). The latter relation is induced by the diagonal actions \(Y_0\mapsto Y_0\Gamma\) and \(Y_\infty\mapsto Y_\infty\Delta\).
\end{theorem}

\begin{proof}
Put \(R=D(\theta)\) and \(S=D(\kappa)\), then under the strong nonresonance conditions \eqref{eq:strong-nonresonance}, there is natural Riemann-Hilbert-Birkhoff bijection
\begin{equation*}
  E_{R,S,x}/\sim_{\mathrm{rat}}
  \;\xrightarrow{\;\sim\;}\;
  F_{R,S,x}/\sim_{\mathrm{diag}}
\end{equation*}
induced by a map \(A(z)\mapsto M(z)\) as in \cite[Theorem~3.7]{ohyama2020space}.
Here \(E_{R,S,x}\) is the space of polynomial coefficients \(A(z)=A_0+z A_1\), with \(A_0\) conjugate to \(R\), \(A_1\) conjugate to \(S\), and roots of \(\det A(z)\) given by \(x\). \(F_{R,S,x}\) is the corresponding space of monodromy data, which consists of matrices \(M(z)\in\Mat_n(\mathcal O(\C^*))\), with \(\sigma_q M=R M(Sz)^{-1}\) and \(\det M(z)\) having only the prescribed simple zeros along the \(q\)-spirals \(x_j q^\Z\) for \(j\in\{1,..,n\}\) and nowhere else.

Now it remains to identify these spaces with our definition of \(E_{\alpha,\beta}(x)\) and \(F_{\alpha,\beta}(x)\) up to the equivalence relations.
Since a constant gauge transformation sends any \(A(z)\in E_{R,S,x}\) to one with the diagonal leading coefficient \(A_1=S\), the inclusion \(E_{\alpha,\beta}(x)\subseteq E_{R,S,x}\) induces an identification
\(E_{\alpha,\beta}(x)/ \sim_{\mathrm{rat}}\simeq E_{R,S,x}/ \sim_{\mathrm{rat}}\).

To identify the monodromy spaces, set \(e_0 = z^{D(\beta)}\) and choose \(e_\infty=z^{D(\alpha)}/\theta_q(-z/X)\). Then the local solutions can be written \(Y_0 = M_0e_0\) and \(Y_\infty = M_\infty e_\infty\) as in \cite{ohyama2020space},
with \(M(z) \coloneqq M_0^{-1}M_\infty=e_0 P(z) e_\infty^{-1}\) satisfying \(\sigma_q M=R M(S z)^{-1}\).
Consequently,
\(P(z)\) is of the form \eqref{eq:birkhoff-entries} precisely when \(M(z)=(h_{ij})\in\Mat_n(\mathcal O(\C^*))\) as in \eqref{eq:birkhoff-holomorphic}.
In this case, by \(\tr D(\alpha)=\tr D(\beta)\), \(\det P(z)\) is single-valued on \(\C^*\) and descents to a meromorphic function on the elliptic curve \(\mathbf E_q=\C^*/q^\Z\). Moreover, we have
\begin{equation*}\label{eq:ORS-dictionary}
  \det M(z)=\det P(z)\prod_{j=1}^n\theta_q(-z/x_j),
\end{equation*}
 so \(\det M(z)\) has the prescribed simple zeros described above precisely when \(\det P(z)\) is a nowhere-vanishing holomorphic function on \(\mathbf E_q\), hence a nonzero constant by compactness of \(\mathbf E_q\).
It follows that the equation \(M(z) = e_0 P(z) e_\infty^{-1}\) identifies \(F_{\alpha,\beta}(x)\) with \(F_{R,S,x}\), and respects the diagonal actions since \(e_0, e_\infty\) are diagonal. Together with
\(E_{\alpha,\beta}(x)/ \sim_{\mathrm{rat}} \simeq E_{R,S,x}/ \sim_{\mathrm{rat}}\), we have
\begin{equation*}\label{eq:RHB-ORS}
  \frac{E_{\alpha,\beta}(x)}{\sim_{\mathrm{rat}}}
  \;\simeq\;
  \frac{E_{R,S,x}}{\sim_{\mathrm{rat}}}
  \;\xrightarrow{\;\sim\;}\;
  \frac{F_{R,S,x}}{\sim_{\mathrm{diag}}}
  \;\simeq\;
  \frac{F_{\alpha,\beta}(x)}{\sim_{\mathrm{diag}}}.
  \qedhere
\end{equation*}
\end{proof}

\section{Construction of Isomonodromy Transformations}
\label{sec:3}

Put \([n] \coloneqq \{1,\dots,n\}\), and fix the local exponents \(\alpha,\beta\in\C^n\). For each \(I\subseteq[n]\), let \(T_I\) send \((x_i,\kappa_i)\) to \((q^{-1}x_i,q\kappa_i)\) for \(i\in I\), while fixing the remaining pairs for \(i\in I^c=[n]\setminus I\). Since \(\alpha,\beta\) are chosen to satisfy \(\tr D(\alpha)=\tr D(\beta)\), we have \(F_{\alpha,\beta}(x)=F_{\alpha,\beta}(T_I x)\).
In this section, we will construct rational transformations \(T_I:E_{\alpha,\beta}(x)\dashrightarrow E_{\alpha,\beta}(T_I x)\) lifting the shifts above, with \(T_\varnothing=\mathrm{id}\).
In Section~\ref{sec:4}, we shall further show that these define a flat family of connection-preserving deformations. 
Consequently, their induced maps on \(E_{\alpha,\beta}(x)/\sim_{\mathrm{rat}}\) make the following diagram commute:
\begin{equation*}
  \begin{tikzcd}[column sep=large,row sep=large]
    E_{\alpha,\beta}(x)/\sim_{\mathrm{rat}}
      \arrow[r,dashed,"T_I"]
      \arrow[d,"\mathrm{RHB}"',"\sim"]
    & 
    E_{\alpha,\beta}(T_I x)/\sim_{\mathrm{rat}}
      \arrow[d,"\sim"',"\mathrm{RHB}"]
    \\
    F_{\alpha,\beta}(x)/\sim_{\mathrm{diag}}
      \arrow[r,"\mathrm{id}"']
    &
    F_{\alpha,\beta}(T_I x)/\sim_{\mathrm{diag}}.
  \end{tikzcd}
\end{equation*}

Now we fix some notations. Let \(\{E_{ij}\}\) be the standard basis of \(\operatorname{gl}_n(\C)\). For \(I\subseteq [n]\), put \(E_I = \sum_{i\in I}E_{ii}\). For \(I,J\subseteq[n]\) and \(V\in\GL(\mathcal V)\), denote by \(V_{I,J}\in\Hom(E_J\mathcal V,E_I\mathcal V)\) the submatrix of \(V\) whose rows and columns indexed by \(I\) and \(J\), respectively. We denote by \(\Delta_I(V) \coloneqq \det V_{I,I}\) the principal minors of \(V\), with \(\Delta_\varnothing(V) = 1\), 
and call \(V\) \emph{principally nonsingular} if all its principal minors are nonzero.

Now for a fixed \(I\subseteq[n]\), we decompose the vector space \(\mathcal V\) by
\begin{equation*}
  \mathcal V=\mathcal V_1\oplus\mathcal V_2,
  \qquad \text{where }
  \mathcal V_1 \coloneqq E_I\mathcal V, \quad
  \mathcal V_2 \coloneqq E_{I^c}\mathcal V.
\end{equation*}
Under this decomposition, we can write \(X,A_1\), and \(V\) in block form as
\begin{equation*}
  X = X_1\oplus X_2, \qquad
  A_1 = D(\kappa) = K_1\oplus K_2, \qquad
  V = \begin{pmatrix}
    V_{11} & V_{12}\\
    V_{21} & V_{22}
  \end{pmatrix},
\end{equation*}
where \(X_a \coloneqq X|_{\mathcal V_a}\in\GL(\mathcal V_a)\), \(K_a \coloneqq A_1|_{\mathcal V_a}\in\GL(\mathcal V_a)\), and \(V_{ab}\in\Hom(\mathcal V_b,\mathcal V_a)\) for \(a,b=1,2\). In particular, the shift operator \(T_I\) induces an action on \(X\) and \(A_1\) by
\begin{equation}\label{eq:parameter-shift}
  T_I X \coloneqq D(T_I x)=q^{-1}X_1\oplus X_2
  \quad \text{and } \quad
  T_I A_1 \coloneqq D(T_I\kappa)=qK_1\oplus K_2,
\end{equation}
while keeping \(q^{D(\alpha)}=-A_1 X\) and hence \(\alpha\) invariant.

\subsection{Lax form and right divisors}

Consider a \(q\)-Lax pair
\begin{equation*}\label{eq:q-Lax pair}
  Y(qz,t)=A(z,t)Y(z,t), \qquad
  Y(z,qt)=B(z,t)Y(z,t).
\end{equation*}
Its compatibility encodes the \(t\)-dynamics of \(A(z,t)\) into the discrete zero-curvature equation
\begin{equation*}\label{eq:zero-curvature}
  A(z,qt) = B(qz,t)A(z,t)B(z,t)^{-1}.
\end{equation*}
Borodin constructed isomonodromy transformations in \cite{borodin2004isomonodromy} by fixing the leading coefficient which is \(A_1\) here, and choosing \(B(z,t)\) as a matrix right divisor of \(A(z,t)\). To deform the selected eigenvalues of \(A_1\) together with their paired roots of \(\det A(z)\) as in \eqref{eq:parameter-shift}, we mark them by the index set \(I\subseteq [n]\) and look for matrix right divisors of the elementary form
\begin{equation}\label{eq:A_I(z)}
  A_I(z)=-z X^{-1}E_I+Q_I, \quad
  \det A_I(z) = \prod_{i\in I}\left(1-\frac{z}{x_i}\right),
\end{equation}
together with their corresponding complementary factors \(A_{\widehat I}(z)=-zX^{-1}E_{I^c}+A_{\widehat I}(0)\). The resulting transformations are then given by the factorizations
\begin{equation}\label{eq:T_I-factorization}
  A(z) = q^{D(\alpha)}A_{\widehat I}(z)A_I(z), \qquad
  \bigl(T_I A\bigr)(z) = A_I(qz)q^{D(\alpha)}A_{\widehat I}(z),
\end{equation}
which naturally send \(x\) to \(T_I x\), and induce compatible deformation equations \((T_IY)(z)=A_I(z)Y(z)\).
We set \(A_\varnothing(z)=Q_\varnothing=\Id_{\mathcal V}\). For \(I=\{i\}\), we denote by \(A_i(z)=-z x_i^{-1}E_i+Q_i\) the elementary matrix divisors.

\subsection{Factorization by the algebraic Riccati equation}

Put \(M \coloneqq -A_1^{-1}A_0=VXV^{-1}\), and write it in block form \(M=(M_{ij})_{1\leq i,j\leq 2}\) with \(M_{ij}\in\Hom(\mathcal V_j,\mathcal V_i)\). For \(I=[n]\), we define \(A_{[n]}\) by chosing its complementary factor \(A_{\widehat{[n]}}=\Id_{\mathcal V}\), so that
\begin{equation}\label{eq:A_[n]}
  A_{[n]}(z) \coloneqq q^{-D(\alpha)}A(z) 
  = X^{-1}\bigl(M - z \Id_\mathcal{V}\bigr),
  \quad \det A_{[n]}(z) = \prod_{i\in [n]}\left(1-\frac{z}{x_i}\right).
\end{equation}

\begin{proposition}\label{prop:Riccati-factorization}
Let \(I\subseteq[n]\), and consider factorizations of the form \(A_{[n]}(z)=A_I^L(z)A_I^R(z)\) into polynomial matrix factors, with 
\begin{equation}\label{eq:A_I^R}
  A_I^R(z) = -z X^{-1}E_I + A_I^R(0), \quad
  \det A_I^R(z) = \prod_{i\in I}\left(1-\frac{z}{x_i}\right).
\end{equation}
Suppose \(\spec(X_1)\cap\spec(X_2)=\varnothing\), then such factorizations exist if and only if \(\Delta_I(V)\neq 0\). In this case, \(A_I^L(z)\) can be chosen to satisfy
\begin{equation}\label{eq:A_I^L}
  A_I^L(z) = -z X^{-1}E_{I^c} + A_I^L(0), \quad
  \det A_I^L(z) = \prod_{j\in I^c}\left(1-\frac{z}{x_j}\right),
\end{equation}
and every such factorization is given by
\begin{equation}\label{eq:A[n]-factorization}
  A_{[n]}(z) = \bigl(L_I(z)S^{-1}\bigr) \bigl(SR_I(z)\bigr), \quad
  S=\begin{pmatrix}
    \Id_{\mathcal V_1} & S_{12}\\
    0 & \Id_{\mathcal V_2}
  \end{pmatrix},
\end{equation}
for a unique \(S_{12}\in\Hom(\mathcal V_2,\mathcal V_1)\), where 
\begin{equation}\label{eq:Riccati-factors}
  \begin{aligned}
    R_I(z) & = -z X^{-1}E_I+
      \begin{pmatrix}
        X_1^{-1}V_{11}X_1V_{11}^{-1} & 0\\
        -V_{21}V_{11}^{-1} & \Id_{\mathcal V_2}
      \end{pmatrix},\\
    L_I(z) & = -z X^{-1}E_{I^c}+
      \begin{pmatrix}
        \Id_{\mathcal V_1} & X_1^{-1}M_{12}\\
        X_2^{-1}V_{21}V_{11}^{-1}X_1 & X_2^{-1}M_{22}
      \end{pmatrix}.
  \end{aligned}
\end{equation}
\end{proposition}

\begin{proof}
The \(I=\varnothing, [n]\) cases follow by definition, so we may assume \(\varnothing\neq I\subsetneq[n]\). 
Suppose \(A_{[n]}(z) = A_I^L(z)A_I^R(z)\) is such a factorization, then \eqref{eq:A_I^R} gives \(\det A_I^R(0)_{22}=1\) by comparing the \(z^{|I|}\)-coefficient, and therefore we may assume \(A_I^R(0)_{22} = \Id_{\mathcal V_2}\) without loss of generality. Now we put \(R_I(z) = S^{-1}A_I^R(z)\) and \(L_I(z) = A_I^L(z)S\) by choosing \(S_{12} = A_I^R(0)_{12}\) in \eqref{eq:A[n]-factorization}, then it reduces \(A_I^R(z)\) to the block lower triangular form
\begin{equation}\label{eq:R_I-form}
  R_I(z) 
  = -z X^{-1}E_I +
    \begin{pmatrix}
        R_{11} & 0\\
        -W & \Id_{\mathcal V_2}
    \end{pmatrix}
  = \begin{pmatrix}
        R_{11}-z X_1^{-1} & 0\\
        0 & \Id_{\mathcal V_2}
    \end{pmatrix}
    \begin{pmatrix}
        \Id_{\mathcal V_1} & 0\\
        -W & \Id_{\mathcal V_2}
    \end{pmatrix}
\end{equation}
for some \(R_{11}\in\End(\mathcal V_1)\) and \(W\in\Hom(\mathcal V_1,\mathcal V_2)\). Inserting it into the factorization gives the polynomial factor
\begin{equation}\label{eq:L_I-reduction}
  L_I(z) 
  = A_{[n]}(z)R_I(z)^{-1}
  = X^{-1}
    \bigl(M - z\Id_{\mathcal V}\bigr)
    \begin{pmatrix}
        \Id_{\mathcal V_1} & 0\\
        W & \Id_{\mathcal V_2}
    \end{pmatrix} 
    \begin{pmatrix}
        R_{11}-z X_1^{-1} & 0\\
        0 & \Id_{\mathcal V_2}
    \end{pmatrix}^{-1}.
\end{equation}
It forces \(R_{11}-z X_1^{-1} = X_1^{-1}(X_1 R_{11} - z\Id_{\mathcal V_1})\) to be a matrix right divisor simultaneously for the two blocks of \(\bigl(M - z \Id_{\mathcal V}\bigr) \left(\begin{smallmatrix}
      \Id_{\mathcal V_1} \\ W
  \end{smallmatrix}\right)\),
so that
\begin{equation}\label{eq:Riccati-blockform}
  \bigl(M - z \Id_{\mathcal V}\bigr) \begin{pmatrix}
      \Id_{\mathcal V_1} \\ W
  \end{pmatrix}
  = \begin{pmatrix}
      \Id_{\mathcal V_1} \\ W
  \end{pmatrix} (X_1R_{11} - z \Id_{\mathcal V_1}),
\end{equation}
by comparing their \(z^1\)-coefficient. Consequently, we have \(X_1 R_{11} = M_{11}+M_{12}W\), with \(W\in\Hom(\mathcal V_1,\mathcal V_2)\) subject to the nonsymmetric algebraic Riccati equation
\begin{equation}\label{eq:algebraic-Riccati-equation}
  M_{21}+M_{22}W-W M_{11}-W M_{12}W = 0.
\end{equation}
The solutions of this equation are in one-to-one correspondence with the \(M\)-invariant complements of \(\mathcal V_2\), by associating \(W\) with its graph 
\(\im\bigl(\begin{smallmatrix}
  \Id_{\mathcal V_1} \\ W
\end{smallmatrix}\bigr)\)
 as in \cite[Theorem~3.3]{freiling2002survey}. 
Since \(R_I(z)=S^{-1}A_I^R(z)\) is block lower triangular as in \eqref{eq:R_I-form}, and has prescribed roots \(\{x_i\}_{i\in I}\) according to \eqref{eq:A_I^R}, we have \(\spec(X_1R_{11})=\spec(X_1)\) disjoint with \(\spec(X_2)\). Hence equation~\eqref{eq:Riccati-blockform} forces the graph of \(W\) to coincide with the \(M\)-invariant subspace spanned by \(\{v_i\}_{i\in I}\), that is,
\begin{equation*}
  \im \begin{pmatrix}
    \Id_{\mathcal V_1} \\ W
  \end{pmatrix} 
  = \im \begin{pmatrix}
    V_{11} \\ V_{21}
  \end{pmatrix}.
\end{equation*}
Such a solution \(W\) exists precisely when \(\Delta_I(V)=\det V_{11}\neq 0\). In this case, it is uniquely given by \(W=V_{21}V_{11}^{-1}\), yielding \(L_I(z)\) in \eqref{eq:Riccati-factors} from \eqref{eq:L_I-reduction} and \eqref{eq:Riccati-blockform}. Moreover, we have \(R_{11} = X_1^{-1}V_{11}X_1V_{11}^{-1}\) for \(R_I(z)\) as in \eqref{eq:Riccati-factors}, by comparing \eqref{eq:Riccati-blockform} with 
\(M\left(\begin{smallmatrix}
  V_{11} \\ V_{21}
\end{smallmatrix} \right) = \left(\begin{smallmatrix}
  V_{11} \\ V_{21}
\end{smallmatrix}\right)X_1\). 
\end{proof}

\subsection{Normalization by the Sylvester equation}

The factorizations in Proposition~\ref{prop:Riccati-factorization}  produce transformed systems
\begin{equation*}
  A_I^R(qz)q^{D(\alpha)}A_I^L(z) = S R_I(qz)q^{D(\alpha)}L_I(z)S^{-1},
\end{equation*}
with a constant gauge freedom \(S\). We normalize this system by requiring its diagonal leading coefficient to be \(q K_1\oplus K_2\), which yields a unique choice of \(S\) determined by a Sylvester equation.

\begin{proposition}\label{prop:Sylvester-normalization}
Let \(I\subseteq[n]\), and assume that factorizations of the form \(A_{[n]}(z)=A_I^L(z)A_I^R(z)\) exist as in Proposition~\ref{prop:Riccati-factorization}.
Suppose \(\spec(q K_1)\cap\spec(K_2) = \varnothing\), then there exists a unique choice \(A_{[n]}(z)=A_{\widehat I}(z)A_I(z)\) among them, for which \(A_I(qz)q^{D(\alpha)}A_{\widehat I}(z)\) has leading term \(z(q K_1\oplus K_2)\). We call \(A_I(z)\) the canonical right divisor associated with \(I\subseteq[n]\). It is obtained by taking \(A_I(z) = S R_I(z)\) and \(A_{\widehat I}(z) = L_I(z)S^{-1}\) as in \eqref{eq:A[n]-factorization} and \eqref{eq:Riccati-factors}, where \(S_{12}\) is chosen the unique solution to the Sylvester equation
\begin{equation}\label{eq:Sylvester-equation}
  S_{12}K_2 - q K_1S_{12} = q X_1^{-1} (A_0)_{12}.
\end{equation}
Consequently, for \(A(z)=q^{D(\alpha)}A_{\widehat I}(z)A_I(z)\), we define the transformation \(T_I\) on it by
\begin{equation}\label{eq:T_I A(z)}
  \bigl(T_I A\bigr)(z) 
  \coloneqq A_I(qz)A(z)A_I(z)^{-1}
  = A_I(qz)q^{D(\alpha)}A_{\widehat I}(z).
\end{equation}
\end{proposition}

\begin{proof}
Suppose \(A_I(z) = S R_I(z)\) and \(A_{\widehat I}(z) = L_I(z)S^{-1}\) for some \(S_{12}\in\Hom(\mathcal V_2,\mathcal V_1)\) as in \eqref{eq:A[n]-factorization} and \eqref{eq:Riccati-factors}. 
Together with \(M = - A_1^{-1} A_0\) we have
\begin{equation*}
  A_I(qz)q^{D(\alpha)}A_{\widehat I}(z) =
  S R_I(0)q^{D(\alpha)}L_I(0)S^{-1} + z \begin{pmatrix}
    q K_1 & -q X_1^{-1} (A_0)_{12} + S_{12}K_2 - q K_1S_{12}\\
    0 & K_2
  \end{pmatrix}.
\end{equation*}
Consequently, it has leading term \(z(q K_1\oplus K_2)\) precisely when \eqref{eq:Sylvester-equation} holds. In this case, there exists a uniquely solution \(S_{12}\) provided that \(\spec(q K_1)\cap\spec(K_2)=\varnothing\).
\end{proof}

\subsection{The transformations \texorpdfstring{\(T_I\)}{TI}}

\begin{theorem}\label{thm:Construction}
Let \(I\subseteq [n]\), and \((V,X)\) a canonical pair of \(A(z)\in E_\alpha(x)\), blocked accordingly, so that
\begin{equation*}
  A(z) = A_0+z(K_1\oplus K_2) 
  = q^{D(\alpha)}X^{-1} \left(V X V^{-1}- z \Id_{\mathcal V}\right)
  \in E_\alpha(x).
\end{equation*}
Suppose \(\spec(X_1)\cap\spec(X_2)=\varnothing\) and \(\spec(qK_1)\cap\spec(K_2)=\varnothing\),
then the transformation 
\begin{equation}\label{eq:T_I-coefficients}
  \bigl(T_I A\bigr)(z) 
  \coloneqq A_I(qz)A(z)A_I(z)^{-1}
  = T_I A_0+z(q K_1\oplus K_2)
\end{equation}
induced by the matrix right divisor \(A_I(z)=-z X^{-1}E_I + Q_I\) exists precisely when \(\Delta_I(V)=\det V_{11}\neq 0\). In this case, we have
\begin{equation}\label{eq:Q_I}
  T_I A_0 = Q_I A_0 Q_I^{-1},\quad
  Q_I = \begin{pmatrix}
    \Id_{\mathcal V_1} & S_{12}\\
    0 & \Id_{\mathcal V_2}
  \end{pmatrix}
  \begin{pmatrix}
    X_1^{-1}V_{11}X_1V_{11}^{-1} & 0\\
    -V_{21}V_{11}^{-1} & \Id_{\mathcal V_2}
  \end{pmatrix}\in\SL(\mathcal V),
\end{equation}
where \(S_{12}\) is the unique solution of \(S_{12}K_2 - q K_1S_{12} = q X_1^{-1} (A_0)_{12}\).

If in addition \(\spec(q^{-1}X_1) \cap \spec(X_2)=\varnothing\), define \(T_I V = (T_I v_1,\dots,T_I v_n)\) columnwise by
\begin{equation}\label{eq:T_I V}
  T_I v_j \coloneqq \begin{cases}
    A_{\widehat I}(q^{-1}x_j)^{-1}q^{-D(\alpha)}v_j,
      & j\in I,\\
    A_I(x_j)v_j,
      & j\notin I,
  \end{cases}
\end{equation}
where \(V=(v_1,\dots,v_n)\in\GL(\mathcal V)\). Then we have \(T_I V\in\GL(\mathcal V)\), with 
\begin{equation}\label{eq:T_I A(z)-parametrization}
  \bigl(T_I A\bigr)(z) = q^{D(\alpha)}D(T_I x)^{-1}
  \left((T_I V)D(T_I x)(T_I V)^{-1} - z\Id_{\mathcal V}\right).
\end{equation}
Consequently, \((T_I A)(z)\in E_\alpha(T_I x)\), and the map \(\bigl(V,D(x)\bigr)\mapsto \bigl(T_I V,D(T_I x)\bigr)\) defines a compatible lift of \(T_I\) to the canonical pair \((V,X)\). In particular, if \(A(z)\in E_{\alpha,\beta}(x)\), then \((T_I A)(z)\in E_{\alpha,\beta}(T_I x)\).
\end{theorem}

\begin{proof}
By Propositions~\ref{prop:Riccati-factorization}, factorizations of the form \(A(z)=q^{D(\alpha)} A_{\widehat{I}}(z)A_I(z)\) exist precisely when \(\Delta_I(V)\neq 0\) under \(\spec(X_1)\cap\spec(X_2)=\varnothing\). Proposition~\ref{prop:Sylvester-normalization} further picks up the unique choice \(A_I(z)=S R_I(z)\) among them under \(\spec(qK_1)\cap\spec(K_2)=\varnothing\), so that \(\bigl(T_I A\bigr)(z)\) is of the form \eqref{eq:T_I-coefficients}. Since \(Q_I = A_I(0)=S R_I(0)\), evaluating \eqref{eq:T_I-coefficients} at \(z=0\) gives \eqref{eq:Q_I}.

Now we assume that \(\spec(q^{-1}X_1) \cap \spec(X_2) = \varnothing\). Together with \(\spec(X_1) \cap \spec(X_2) = \varnothing\), we have \(A_{\widehat{I}}(q^{-1}x_i)\) and \(A_I(x_j)\) invertible for \(i\in I\) and \(j\in I^c\) by their determinants in \eqref{eq:A_I^L} and \eqref{eq:A_I^R} respectively, thus \eqref{eq:T_I V} is well-defined, with each \(T_I v_j\neq 0\).
Since \(v_j\in\ker A(x_j)\) for \(j\in [n]\), substituting \eqref{eq:T_I V} into \eqref{eq:T_I A(z)} gives
\begin{equation*}
  \bigl(T_I A\bigr)(T_I x_j)T_I v_j = \begin{cases}
    A_I(x_j)v_j=0, & j\in I,\\
    A_I(qx_j)A(x_j)v_j=0, & j\notin I,
  \end{cases}
\end{equation*}
so that \((T_I v_j, T_I x_j)\) are the eigenpairs of the matrix \(-(T_I A_1)^{-1} T_I A_0\). Now by \(\spec(q^{-1}X_1) \cap \spec(X_2) = \varnothing\), each of its eigenspaces is the image of an eigenspace of \(-A_1^{-1}A_0\) under \eqref{eq:T_I V}, with \(A_{\widehat I}(q^{-1}x_j)^{-1}q^{-D(\alpha)}\) and \(A_I(x_j)\in \GL(\mathcal{V})\) preserving their dimension. Consequently, we have \(T_I V\in\GL(\mathcal V)\), so that \((T_I V,D(T_I x))\) is a canonical pair of \(T_I A(z)\) as in \eqref{eq:T_I A(z)-parametrization}, yielding \((T_I A)(z)\in E_\alpha(T_I x)\). In particular, by \(T_I A_0 = Q_I A_0 Q_I^{-1}\), it preserves the conjugacy class of the constant term \(A_0 = C_0 q^{D(\beta)} C_0^{-1}\).
\end{proof}

\begin{definition}\label{def:T_I-domain}
Let \(I\subseteq[n]\), and suppose \(\spec(X_1)\cap\spec(X_2)=\varnothing\). Then we define
\begin{equation*}
  \begin{aligned}
  E_\alpha^{I}(x)
  & \coloneqq
    \left\{ A(z)\in E_\alpha(x)
      \;\middle|\;
      \Delta_I(V)\neq 0 \text{ for a pair } (V,X) \text{ of }A(z)
    \right\},\\
  E_{\alpha,\beta}^{I}(x)
  & \coloneqq 
    E_\alpha^{I}(x)\cap E_{\alpha,\beta}(x).
  \end{aligned}
\end{equation*}
\end{definition}

Here \(\spec(X_1)\cap\spec(X_2)=\varnothing\) makes the condition \(\Delta_I(V)\neq 0\) independent of \(V\). Thus \(E_\alpha^{I}(x)\) is Zariski open in \(E_\alpha(x)\), nonempty (take \(V=\Id_{\mathcal V}\)), and hence dense because \(E_\alpha(x)\) is irreducible. Under the three spectral assumptions in Theorem~\ref{thm:Construction}, the construction gives a regular map \(T_I:E_\alpha^{I}(x)\longrightarrow E_\alpha(T_I x)\), which preserves the conjugacy class of \(A_0\) and therefore descents to \(T_I:E_{\alpha,\beta}^{I}(x)\longrightarrow E_{\alpha,\beta}(T_I x)\).

In particular, \((T_{[n]}A)(z) = q^{-D(\alpha)}A(qz)q^{D(\alpha)}\) is defined for all \(A(z)\in E_\alpha(x)\), with compatible lift \(T_{[n]}(V,X) = (q^{-D(\alpha)}V,q^{-1}X)\). It restricts to an isomorphism
\begin{equation}\label{eq:T[n]-isomorphism}
  T_{[n]}: E_\alpha^{I}(x) 
  \xrightarrow{\sim} E_\alpha^{I}(q^{-1}x),
\end{equation}
which descents to \(E_{\alpha,\beta}^{I}(x) \xrightarrow{\sim} E_{\alpha,\beta}^{I}(q^{-1}x)\), for every \(I\subseteq[n]\).

\section{Composition and Flatness}
\label{sec:4}

\subsection{Composition law}

For disjoint subsets \(I,J\subseteq[n]\), we compare the composition \(T_J T_I\) with \(T_{I\sqcup J}\). Throughout this subsection, we make the following \textbf{\textit{spectral assumptions}}: whenever \(T_I\) acts on \(A(z)=D(\kappa)(z\Id_{\mathcal V}-VXV^{-1}) \in E_\alpha(x)\) for some \(I\subseteq[n]\), the corresponding parameters satisfy
\begin{equation*}\label{eq:spectral-assumptions}
  x_i\neq x_j, \quad
  q\kappa_i\neq\kappa_j, \quad
  q^{-1}x_i\neq x_j \quad
  \text{for } i\in I,\ j\in I^c.
\end{equation*}

\begin{proposition}\label{prop:minor-shift}
Let \(I,J\subseteq[n]\) disjoint. Under the spectral assumption, \(T_I V\) is defined and invertible precisely when \(\Delta_I(V)\neq 0\). In this case,
\begin{equation}\label{eq:schur-complement}
  (T_I V)_{J,J} = V_{J,J}-V_{J,I}(V_{I,I})^{-1}V_{I,J}.
\end{equation}
Consequently, we have
\begin{equation}\label{eq:minor-shift}
  \Delta_J(T_I V) = \frac{\Delta_{I\sqcup J}(V)}{\Delta_I(V)}.
\end{equation}
\end{proposition}

\begin{proof}
\(T_I V\) is defined and invertible by Theorem~\ref{thm:Construction} under the spectral assumptions. By equation~\eqref{eq:T_I V} we have \((T_I V)e_j=A_I(x_j)Ve_j\) for \(j\in I^c\), where \(A_I(z) = -z X^{-1}E_I+Q_I\). Taking the \(I^c\times I^c\) block and using the formula \eqref{eq:Q_I} for \(Q_I\), we obtain
\begin{equation*}
  (T_I V)_{I^c,I^c} 
  = (Q_I V)_{I^c,I^c} 
  = V_{I^c,I^c} - V_{I^c,I}V_{I,I}^{-1}V_{I,I^c}.
\end{equation*}
Restricting this identity to \(J\subseteq I^c\) gives \eqref{eq:schur-complement}, and the Schur-complement formula yields \eqref{eq:minor-shift}.
\end{proof}

For \(T_I A(z)\), we denote by \(T_I A_J(z) = - z X^{-1}E_J + T_I Q_J\) its right divisor defining \(T_J T_I A(z)\).

\begin{theorem}[Composition law]\label{thm:Composition}
Let \(I,J\subseteq[n]\) disjoint, and suppose \((V,X)\) is a canonical pair of \(A(z)\in E_\alpha(x)\). Then under the spectral assumption, \(T_J T_I A(z)\) and \(T_{I\sqcup J}A(z)\) are defined in \(E_\alpha(T_{I\sqcup J}x)\) if and only if
\begin{equation*}
  \Delta_I(V)\neq 0, \qquad
  \Delta_{I\sqcup J}(V)\neq 0.
\end{equation*}
In this case, we have
\begin{equation}\label{eq:composition}
  T_J T_I A(z) = T_{I\sqcup J} A(z), \qquad
  T_J T_I (V,X)=T_{I\sqcup J}(V,X).
\end{equation}
\end{theorem}

\begin{proof}
Under the spectral assumptions, Theorem~\ref{thm:Construction} defines \(T_I A(z)\in E_\alpha(T_I x)\) and \(T_{I\sqcup J}A(z) \in E_\alpha(T_{I\sqcup J}x)\) precisely when \(\Delta_I(V), \Delta_{I\sqcup J}(V)\neq 0\). In this case, we have \(\Delta_J(T_I V)\neq 0\) by \eqref{eq:minor-shift}, and it defines \(T_J T_I A(z)\in E_\alpha(T_{I\sqcup J})(x)\).

Set \(B_{I,J}(z) \coloneqq \bigl(T_I A_J\bigr)(z)A_I(z)\). It sends \(A(z)\) to \(T_J T_I A(z)\in E_\alpha(T_{I\sqcup J}x)\) and hence satisfies
\begin{equation}\label{eq:B_IJ-lax-identity}
  \bigl(T_J T_I A\bigr)(z)B_{I,J}(z)
  = B_{I,J}(qz)A(z).
\end{equation} 
Write \(A_I(z) = - z X^{-1}E_I + Q_I\), \(T_I A_J(z)= - z X^{-1}E_J + T_I Q_J\). Since \(I\cap J=\varnothing\), equation~\eqref{eq:Q_I} gives \(E_J Q_I = E_J + E_J Q_I E_I\), and similarly \((T_I Q_J) E_I = E_I + E_J (T_I Q_J) E_I\). Therefore, we have
\begin{equation*}\label{eq:B_IJ-coefficient}
  B_{I,J}(z) = B_{I,J}(0)-z\bigl(X^{-1}E_{I\sqcup J} + N\bigr), \quad
  N=E_J \bigl(X^{-1}Q_I + (T_I Q_J)X^{-1}\bigr) E_I.
\end{equation*}
Comparing the \(z^2\)-coefficient in \eqref{eq:B_IJ-lax-identity} then yields \([A_1,N]=0\), so that \(N = 0\) by the spectral assumptions \(\kappa_i \neq \kappa_j\) \((i\in I, j\in J)\) for \(T_J\) acting on \(T_I A\). Hence for \(I\cap J=\varnothing\), \(B_{I,J}(z)\) has the prescribed form and determinant
\begin{equation}\label{eq:B_IJ-form}
  B_{I,J}(z) = -z X^{-1}E_{I\sqcup J} + B_{I,J}(0), \qquad
  \det B_{I,J}(z) = \prod_{k\in I\sqcup J}\left(1-\frac{z}{x_k}\right).
\end{equation}

Now it suffices to show that \(B_{I,J}(z)\) is a indeed matrix right divisor of \(A_{[n]}(z)\). Set \(C_{I,J}(z)\coloneqq A_{[n]}(z)B_{I,J}(z)^{-1}\) as the complementary factor of \(B_{I,J}(z)\), and denote by \(T_I A_{\widehat J}(z)\) the complementary factor of \(T_I A_J(z)\). Then comparing the two factorizations
\begin{equation*}
  \bigl(T_I A\bigr)(z) = A_I(qz)q^{D(\alpha)}A_{\widehat I}(z)
  = q^{D(\alpha)} \bigl(T_I A_{\widehat J}\bigr)(z) \bigl(T_I A_J\bigr)(z)
\end{equation*}
gives the following identities for \(C_{I,J}(z)\) as matrix-valued rational functions in \(z\):
\begin{equation}\label{eq:C_IJ}
  C_{I,J}(z) 
  = A_{\widehat I}(z)\bigl(T_I A_J\bigr)(z)^{-1}
  = q^{-D(\alpha)}A_I(qz)^{-1} q^{D(\alpha)} \bigl(T_I A_{\widehat J}\bigr)(z).
\end{equation}
Here the first expression can have poles in \(\C\) only at the zeros \(x_j\) of \(\det(T_I A_J)\) for \(j\in J\), while the second is holomorphic there because \(\det A_I(q x_j) = \prod_{i\in I}(1-q x_j/x_i)\neq 0\) by the spectral assumption for \(T_I\), thus \(C_{I,J}(z)\) is holomorphic on \(\C\).
Moreover, since \(A_{\widehat I}(z)\) has degree at most one and equation~\eqref{eq:A_I(z)} gives \((T_I A_J)(z)^{-1}=O(1)\) as \(z\to\infty\), we have \(C_{I,J}(z)=O(z)\), so that \(C_{I,J}(z)\) is a polynomial matrix of degree at most one. 
Consequently, \(B_{I,J}(z)\) is a right divisor of \(A_{[n]}(z)\). Hence by Proposition~\ref{prop:Riccati-factorization} and~\ref{prop:Sylvester-normalization}, the conditions~\eqref{eq:B_IJ-lax-identity} and~\eqref{eq:B_IJ-form} yield uniquely
\begin{equation}\label{eq:B_IJ}
  B_{I,J}(z) \coloneqq \bigl(T_I A_J\bigr)(z)A_I(z) = A_{I\sqcup J}(z),\qquad C_{I,J}(z)=A_{\widehat{I\sqcup J}}(z).
\end{equation}
This proves that \(T_J T_I A(z) = T_{I\sqcup J}A(z)\). Columnwise computation with~\eqref{eq:T_I V} then gives
\begin{equation*}\label{eq:T_IJ V}
  T_J T_I v_k \coloneqq \begin{cases}
    \bigl(T_I A_J\bigr)(q^{-1}x_k)A_{\widehat I}(q^{-1}x_k)^{-1}q^{-D(\alpha)}v_k,
      & k\in I,\\
    \bigl(T_I A_{\widehat J}\bigr)(q^{-1}x_k)^{-1}q^{-D(\alpha)}A_I(x_k)v_k,
      & k\in J,\\
    \bigl(T_I A_J\bigr)(x_k)A_I(x_k)v_k,
      & k\notin I\sqcup J,
  \end{cases}
\end{equation*}
which coincides with \(T_{I\sqcup J}V\) by noticing \eqref{eq:C_IJ} and \eqref{eq:B_IJ}.
\end{proof}

\begin{corollary}[Flatness]\label{cor:flatness}
Let \(I,J\subseteq[n]\) be disjoint. Then under the spectral assumptions, the actions of \(T_J T_I\) and  \(T_I T_J\) are simultaneously defined on \(A(z)\) precisely when \(\Delta_I(V)\Delta_J(V)\Delta_{I\sqcup J}(V)\neq 0\). In this case, their action on \(A(z)\) and its canonical pair \((V,X)\) coincide with each other.
\end{corollary}

\begin{proof}
By applying Theorem~\ref{thm:Composition} to \(T_J T_I\) and \(T_I T_J\), both of them coincide with \(T_{I\sqcup J}\).
\end{proof}

\begin{corollary}[Birationality]\label{cor:birationality}
Under the spectral assumption, each \(T_I\) (for \(I\subseteq [n]\)) restricts to an isomorphism of Zariski-open subvarieties
\begin{equation*}
  T_I|_{E_\alpha^I(x)}: E_\alpha^{I}(x) \xrightarrow{\sim} E_\alpha^{I^c}(T_I x),
\end{equation*}
with its inverse \(T_I^{-1}=T_{[n]}^{-1}T_{I^c}\). In particular, \(T_I:E_\alpha(x)\dashrightarrow E_\alpha(T_I x)\) is a birational map and restricts to an isomorphism \(E_{\alpha,\beta}^{I}(x)\xrightarrow{\sim} E_{\alpha,\beta}^{I^c}(T_I x)\).
\end{corollary}

\begin{proof}
The rational map \(T_I:E_\alpha(x)\dashrightarrow E_\alpha(T_I x)\) is regular on the dense open subset \(E_\alpha^{I}(x)\subseteq E_\alpha(x)\) as in Theorem~\ref{thm:Construction}. Since \(\Delta_{I^c}(T_I V)=\Delta_{[n]}(V)/\Delta_I(V)\neq 0\), it takes value in the dense open subset \(E_\alpha^{I^c}(T_I x)\subseteq E_\alpha(T_I x)\). By Theorem~\ref{thm:Composition}, we have on these subsets \(T_{I^c}T_I = T_{[n]}\) and \(T_I T_{I^c}=T_{[n]}\) which are isomorphisms by~\eqref{eq:T[n]-isomorphism}, and this gives \(T_I^{-1}=T_{[n]}^{-1}T_{I^c}\) as a regular inverse of \(T_I\) on \(E_\alpha^{I^c}(T_I x)\). Consequently, \(T_I|_{E_\alpha^I(x)}\) is an isomorphism. Since it preserves the conjugacy class of \(A_0\), the restricted map \(E_{\alpha,\beta}^{I}(x)\rightarrow E_{\alpha,\beta}^{I^c}(T_I x)\) is also an isomorphism.
\end{proof}

\begin{corollary}[Local tau function]\label{cor:local-tau-function}
Let \((V(x),D(x))\) be a flat canonical pair for some \(A(z;x)\in E_\alpha(x)\), in the sense that \(T_I \bigl(V(x),D(x)\bigr) = \bigl(V(T_I x),D(T_I x)\bigr)\) whenever \(T_I A(z;x)\) is defined. Then the equations
\begin{equation}\label{eq:local-tau-function}
  \frac{\tau(T_I x)}{\tau(x)} = \Delta_I\bigl(V(x)\bigr), 
  \quad \text{for } I \subseteq [n],\ \Delta_I(V(x))\neq 0
\end{equation}
are compatible and locally define a function \(\tau(x)\in\C^*\), unique up to a nonzero constant. In particular, if \(V(x)\) is principally nonsingular, then its characteristic polynomial is given by
\begin{equation}\label{eq:tau-characteristic-polynomial}
  \det \bigl(\lambda\Id_{\mathcal V}-V(x)\bigr)
  = \frac{1}{\tau(x)} 
  \sum_{I\subseteq[n]} (-1)^{|I|}\lambda^{n-|I|}\tau(T_I x).
\end{equation}
\end{corollary}

\begin{proof}
For disjoint \(I,J\subseteq[n]\), Proposition~\ref{prop:minor-shift} gives the following identity
\begin{equation*}\label{eq:minor-compatibility}
  \Delta_J(V(T_I x)) \Delta_I(V(x)) 
  = \Delta_{I\sqcup J}(V(x))
  = \Delta_I(V(T_J x)) \Delta_J(V(x)),
\end{equation*}
which is precisely the compatibility condition for \(T_J T_I\tau(x) = T_{I\sqcup J}\tau(x) = T_I T_J\tau(x) \). Therefore equation~\eqref{eq:local-tau-function} defines the function \(\tau(x)\) locally, unique up to a chosen initial value. Equation~\eqref{eq:tau-characteristic-polynomial} follows from the principal-minor expansion of characteristic polynomial.
\end{proof}

\subsection{Canonical solutions and isomonodromy}

Now we show that the \(T_I\)'s are indeed connection-preserving deformations. For simplicity, we assume the stong nonresonance conditions~\eqref{eq:strong-nonresonance} as in Theorem~\ref{thm:RHB-correspondence}.

\begin{theorem}[Isomonodromy]\label{thm:Isomonodromy}
Let \(A(z;x) \in E_{\alpha,\beta}(x)\) be a family of stongly non-resonant system coefficients parametrized by the position of \(x\), with canonical solutions of the form
\begin{equation}\label{eq:Y(z;x)}
  Y_0(z;x) = C_0(x)Y_0^*(z;x)z^{D(\beta)}, \qquad
  Y_\infty(z;x) = Y_\infty^*(z;x)\frac{z^{D(\alpha)}}{(z X^{-1};q)_\infty},
\end{equation}
as in Proposition~\ref{prop:canonical-solutions}. Then for \(I\subseteq [n]\), 
we have \(P(z;T_I x) = P(z;x)\) if and only if 
\begin{equation}\label{eq:T_I C_0}
  A(z;T_I x) = T_I A(z;x) \coloneqq A_I(qz;x)A(z;x)A_I(z;x)^{-1}, \qquad C_0(T_I x)=Q_I C_0(x).
\end{equation}
\end{theorem}

\begin{proof}
Let \(T_I Y(z) \coloneqq A_I(z) Y(z)\) for \(Y = Y_0,Y_\infty\) respectively, then by~\eqref{eq:T_I C_0} they are solutions of the deformed equation \(Y(q z) = T_I A(z) Y(z)\), and can be  written in the form
\begin{equation*}\label{eq:A_I Y(z)}
  A_I(z)Y_0(z) 
    = (T_I C_0)\widetilde Y_0^*(z)z^{D(\beta)}, \qquad 
  A_I(z)Y_\infty(z) 
    = \widetilde Y_\infty^*(z)
    \frac{z^{D(\alpha)}}{(z D(T_I x)^{-1};q)_\infty},
\end{equation*}
where \(\widetilde Y_0^*(z) = (T_I C_0)^{-1}A_I(z)C_0 Y_0^*(z)\) and \(\widetilde Y_\infty^*(z) = A_I(z)Y_\infty^*(z) (\Id_{\mathcal V}-z X^{-1}E_I)^{-1}\).

To show that \(T_I\) is connection preserving, it suffices to check that these are the canonical solutions normalized as in Proposition~\ref{prop:canonical-solutions}.
It is straightforward to see that \(\widetilde Y_0^*(z)\) is holomorphic at \(z=0\), with \(\widetilde Y_0^*(0) = (T_I C_0)^{-1}Q_I C_0\). Hence it satisfies the normalization~\eqref{eq:normalization} precisely when \(T_I C_0 = Q_I C_0\). For \(z=\infty\), expanding \(Y_\infty^*(z)=I+U_1 z^{-1}+O(z^{-1})\), \((\Id_{\mathcal V}-z X^{-1}E_I)^{-1}=\diag\{- z^{-1}X_1 + O(z^{-2}),\Id_{\mathcal V_2}\}\), and \(A_I(z)=S R_I(z)\) as in Proposition~\ref{prop:Sylvester-normalization} gives
\begin{equation*}\label{eq:A_I Y-normalization}
  \widetilde Y_\infty^*(z) 
  = \begin{pmatrix}
    \Id_{\mathcal V_1} & S_{12}-X_1^{-1}(U_1)_{12}\\
    0 & \Id_{\mathcal V_2}
  \end{pmatrix} + O\left(\frac{1}{z}\right).
\end{equation*}
Inserting \(Y_\infty^*(z)\) into~\eqref{eq:Y_infty^*} shows that \((X^{-1}U_1)A_1 - q A_1 (X^{-1}U_1) = q X^{-1}(A_0 - q^{D(\alpha)})\), thus \((X^{-1} U_1)_{12}=X_1^{-1}(U_1)_{12}\) coincides with \(S_{12}\) by~\eqref{eq:Sylvester-equation}. Therefore \(\widetilde Y_\infty^*(z)\) is holomorphic at \(z=\infty\), and has the required normalization \(\widetilde Y_\infty^*(\infty) = \Id_{\mathcal V}\). Consequently, \eqref{eq:T_I C_0} keeps both \(T_I Y_0\) and \(T_I Y_\infty\) canonical, and therefore preserves the connection matrix \(P(z;x)\).

Conversely, suppose \(P(z;T_I x) = P(z;x)\), then by \(P(z;x)=Y_0^{-1}(z;x)Y_\infty(z;x)\) as in~\eqref{eq:birkhoff-connection-matrix}, we have
\begin{equation*}
  B(z;x) 
  \coloneqq Y_\infty(z;T_I x) Y_\infty(z;x)^{-1}
  = Y_0(z;T_I x) Y_0(z;x)^{-1}.
\end{equation*}
Now write \(Y_\infty(z;x)=\bar Y_\infty^*(z;x)\frac{z^{D(\alpha)}(q;q)_\infty}{\theta_q(-z X^{-1})}\), then by equation~\eqref{eq:Y_infty^*}, we have \(\bar Y_\infty^*(z;x) = Y_\infty^*(z;x)(q X/z;q)_\infty\) and it satisfies 
\begin{equation*}\label{eq:Ybar}
  \bar Y_\infty^*(q z;x)\bigl(z A_1\bigr) 
  = \bigl(A_0+z A_1\bigr) \bar Y_\infty^*(z;x),
\end{equation*}
Therefore it is holomorphic on \(\C^*\), with its inverse \(\bar Y_\infty^*(z;x)^{-1}\) having at most simple poles contained in \(\{x_j q^m:m\ge 1\}\) for \(j\in [n]\). In particular, these are the only possible poles of
\begin{equation*}\label{eq:B(z;x)}
  B(z;x)
  = \bar Y_\infty^*(z;T_I x)\begin{pmatrix}
    - z X_1^{-1} & \\ & \Id_{\mathcal V_2}
  \end{pmatrix} \bar Y_\infty^*(z;x)^{-1}
  = C_0(T_I x)Y_0^*(z;T_I x) Y_0^*(z;x)^{-1}C_0(x)^{-1},
\end{equation*}
where the second identity shows that \(B(z;x)\) is holomorphic there. Consequently, \(B(z;x)\) is holomorphic on \(\C\), and therefore has the same form \(B(z;x) = - z X^{-1} E_I +B(0;x)\) and determinant as in~\eqref{eq:A_I(z)} by~\eqref{eq:normalization} and \eqref{eq:det Y(z)}. 
Now we put \(C(z)= B(z;x)A_I(z;x)^{-1}\), then it satisfies
\begin{equation}\label{eq:C(z)}
  C(q z) = \bigl(B[A](z)\bigr) C(z) \bigl(T_I A(z)\bigr)^{-1},
  \quad \text{where } B[A](z) \coloneqq B(q z;x)A(z;x)B(z;x)^{-1}.
\end{equation}
By inverting \(A_I(z) = S R_I (z)\) with~\eqref{eq:R_I-form}, one can verify that \(C(z)\) is holomorphic at both \(z=0,\infty\). This allows us to locate its possible poles, using the roots of either \(\det B[A](z)\) or \(\det T_I A(z)\). Since \(\det B[A](z) = \det T_I A(z)\), euqation~\eqref{eq:C(z)} forces \(C(z)\) to have no poles in \(\C^*\) under the stong nonresonance conditions, in the same spirit of Borodin's uniqueness proof in \cite{borodin2004isomonodromy}. Consequently \(C(z)\) is a constant matrix by Liouville's Theorem, hence equal to \(\Id_{\mathcal V}\) by Proposition~\ref{prop:Sylvester-normalization} and the normalization of \(Y_\infty(z;T_I x)\).
\end{proof}

\section{Continuous Limit}
\label{sec:5}

Our construction reduces in its \(q\to 1\) limit to the isomonodromic deformation of Jimbo, Miwa, and Ueno \cite{jimboMiwaUeno1981}. Put \(\varepsilon=q-1\), and let \(u=(u_1,\dots,u_n)\in(\C^*)^n\) have pairwise distinct entries. Consider a family of coefficients \(A(z;u,\varepsilon)\) such that
\begin{equation*}
  A_0 = \Id_{\mathcal V} + \varepsilon A, \qquad
  A_1 = \varepsilon U, \qquad
  U=D(u).
\end{equation*}
We suppress the dependence on \(u\) and \(\varepsilon\) and write \(T_i \coloneqq T_{\{i\}}\).

\begin{proposition}
\label{prop:formal-continuous-limit}
As \(\varepsilon\to 0\), the \(q\)-difference system and its deformation equations formally reduce to
\begin{equation}\label{eq:limiting-lax-pair}
  \frac{\partial Y}{\partial z} 
    = \left(U+\frac{A}{z}\right)Y, \qquad
  \frac{\partial Y}{\partial u_i} 
    = \left(z E_i+\ad_U^{-1}[E_i,A] \right)Y \quad 
  (i\in[n]).
\end{equation}
In particular, their compatibility condition gives
\begin{equation}\label{eq:limiting-isomonodromy-equation}
  \frac{\partial A}{\partial u_i} 
    = \left[ \ad_U^{-1}[E_i,A],A \right] \quad (i\in[n]),
\end{equation}
where \(\ad_U^{-1}\) denotes the inverse of \(\ad_U\) on the off-diagonal matrices.
\end{proposition}

\begin{proof}
The \(q\)-difference equation reads \(\frac{Y(qz)-Y(z)}{\varepsilon z} = \left(U+\frac{A}{z}\right)Y(z)\), yielding \(\frac{\partial Y}{\partial z}=\left(U+\frac{A}{z}\right)Y\) as its formal \(\varepsilon\to 0\) limit. 
The canonical pair \((V,X)\) of the system is determined by
\begin{equation*}
   V X V^{-1} = -A_1^{-1}A_0
   = - \varepsilon^{-1} U^{-1}(\Id_{\mathcal V}+\varepsilon A),
\end{equation*}
therefore by the standard perturbation theories, one can take the eigenpairs as
\begin{equation*}
  x_j^{-1} = - \varepsilon u_i  + O(\varepsilon^2), \quad
  v_{j} = e_j + \varepsilon \sum_{i\neq j} e_i \frac{u_j A_{ij}}{u_i-u_j} + O(\varepsilon^2), 
  \qquad \text{for } j\in [n].
\end{equation*}
Inserting them into \eqref{eq:Q_I} with \(I=\{i\}\) gives
\begin{equation*}\label{eq:limiting-coefficients}
  A_i(z) = \varepsilon u_i z E_i + Q_i, \quad 
  Q_i = \Id_{\mathcal V} + \varepsilon u_i\ad_U^{-1}[E_i,A] 
      + O(\varepsilon^2).
\end{equation*}
Since \(T_i u_i=q u_i=u_i+\varepsilon u_i\), comparing the \(\varepsilon^1\)-coefficient in \(T_i Y(z)=A_i(z)Y(z)\) and \(T_i A_0=Q_i A_0 Q_i^{-1}\) gives the equations for \(\partial Y/\partial u_i\) and \(\partial A/\partial u_i\), respectively.
\end{proof}

\section{Monodromy Problem}
Since the connection matrix \(P(z)\) has constant determinant, it put a strong constraint on the characteristic constants. A central problem is to determine these constants explicitly. Consider a flat family of system coefficients \(A^{(n)}(z,t)\) of the form \(A^{(n)}(z,t)=A^{(n)}_0(t) + z A^{(n)}_1(t)\), with 
\begin{equation*}
  A^{(n)}_0(t)=C^{(n)}_0(t)D(\lambda^{(n)})C^{(n)}_0(t)^{-1},
  \quad
  A^{(n)}_1(t)=\diag\{t A^{(n-1)}_1, \kappa_n\},
\end{equation*}
where \(\lambda^{(n)}=(\lambda^{(n)}_1,\dots,\lambda^{(n)}_n)\) and \(A^{(n-1)}_1=D(\kappa_1,\dots,\kappa_{n-1})\).
Such a family is typically characterized by its \(t\to 0\) boundary behavior. More precisely, given a expected boudnary value \(\Omega = (\Omega_{ij})\in \GL(\C^{n-1}\oplus\C)\), write \(\lambda^{(n-1)}=(\lambda^{(n-1)}_1,\dots,\lambda^{(n-1)}_{n-1})\) for the spectrum of \(\Omega_{11}\), one can expect under approciate spectrum condition on \(\lambda^{(n-1)}\) that there exists a unique solution \(A^{(n)}(z,t)\) to the flat deformation in \(t\), which has the prescribed limit under a gauge transformation:
\begin{equation*}
  \lim_{q^\Z\ni t\to 0}Q_\Omega(t)A^{(n)}(z,t)Q_\Omega(t)^{-1} = \kappa_n E_n z + \Omega,
  \quad \text{for } Q_\Omega(t)=\diag\{t^{\Omega_{11}}t^{D(\alpha_1,\dots,\alpha_{n-1})},1\}.
\end{equation*}
Consequently, its solutions \(Q_\Omega(t)Y_0(z,t)\) and \(Q_\Omega(t)Y_\infty(z,t)\) converges to a solution of the limiting system, which has \(q\)-Stokes phenomonon analized in \cite{linMaXu2024explicit}. By dianolizing the upper-left block \(\lambda^{(n-1)}=C_\Omega^{-1}\Omega_{11}C_\Omega\), we can view \(C_\Omega\in\GL(\C^{n-1})\subseteq\GL(\C^n)\) and write \(\Lambda_{ij}\) for the entries of \(\Lambda=C_\Omega^{-1}\Omega C_\Omega\). Then by analyzing the \(q\)-Stokes phenomonon and matching it with the poles of the solutions, it is conjectured the following, whose main difficulty in its proof lies in a uniform \(q\)-Borel-Laplace analysis of the degeneration process of \(Q_\Omega(t)Y_\infty(z,t)\) as \(t\to 0\), where the series itself becomes divergent and the \(q\)-Stokes phenomonon arises.
\begin{conjecture}
The characteristic constants of \(A^{(n)}(z,t)\) above are given by a recursion relation 
\begin{equation}
  p^{(n)}_{ij} = \begin{cases}
  \sum_{k = 1}^{n - 1} 
  \frac{\Lambda_{n k}}{\lambda_i^{(n)}} 
  \cdot
  \frac{
    \bigl(q\lambda_{\,\widehat{i}}^{(n)} / \lambda^{(n - 1)}_k, \lambda^{(n - 1)}_{\,\widehat{k}} /\lambda^{(n)}_i ; q\bigr)_{\infty}
    }{
    \bigl(q \lambda_{\,\widehat{k}}^{(n - 1)} / \lambda^{(n - 1)}_k, \lambda_{\,\widehat{i}}^{(n)} / \lambda^{(n)}_i ; q\bigr)_{\infty}
  } 
  \cdot
  \frac{
    \theta_q \bigl( - \kappa_j \lambda^{(n)}_i/ (\kappa_n \lambda^{(n - 1)}_k) \bigr)
    }{
    \theta_q \bigl( - \kappa_j/\kappa_n \bigr)
  } 
  \cdot p^{(n - 1)}_{k j}, 
    & j<n, \\[10pt]
  \frac{
      \bigl(\lambda^{(n - 1)} / \lambda^{(n)}_i ;q\bigr)_{\infty}
    }{
      \bigl(\lambda^{(n)}_{\,\widehat{i}} / \lambda^{(n)}_i ; q\bigr)_{\infty}
  }, & j=n.
  \end{cases}
\end{equation}
where \(p^{(n-1)}_{kj}\) are the characteristic constants for the system with \(A^{(n-1)}(z) = z A^{(n-1)}_1 + \Omega_{11}\).
Moreover, for generic \(P(z)\in F_{\alpha,\beta}(x)\), its characteristic constants can be given this way. 
\end{conjecture}

\section*{Appendix: An Explicit Rank-Two Example}
\label{Appendix}
It is known that for \(n=2\), the system~\eqref{eq:system} is solvable in terms of Heine's basic hypergeometric series \cite{manoAsymptoticBehaviourBoundary2010,ohyamaAnalyticSolutions2009}. Our construction is very explicit in this case. For \({\theta_1}/{\theta_2}, {\kappa_1}/{\kappa_2}, {x_1}/{x_2}\in \C^*\backslash q^{\mathbb Z}\), take \(q^{\alpha_i}= -\kappa_i x_i\), \(q^{\beta_i}=\theta_i\) and \(t = x_2/x_1\), with \(\alpha_1+\alpha_2=\beta_1+\beta_2\).
We define the function 
\begin{equation}\label{eq:tau(x1,x2)}
  \tau(x_1,x_2) 
  \coloneqq \tau_* \mathcal G_{\alpha_1}(x_1) \mathcal G_{\alpha_2}(x_2) 
  \frac{(q^{\alpha_1-\beta_1+1}t;q,q)_\infty
    (q^{\alpha_1-\beta_2+1}t;q,q)_\infty}
    {(qt;q,q)_\infty (q^{\alpha_1-\alpha_2+1}t;q,q)_\infty},
\end{equation}
where \(\mathcal G_\gamma(x) \coloneqq \exp 
\left(- \frac{\gamma(\log x)^2}{2\log q} - \frac{\gamma}{2}\log x\right)\), \((z;q,q)_\infty \coloneqq \prod_{r,s\geq0}(1-zq^{r+s})\) and \(\tau_*\in\mathbb C^*\). 
Write \(\tau =\tau(x)\), \(\tau_i =\tau(T_i x)\), and \(\tau_{12} =\tau(T_1 T_2 x)\), then by 
\(\mathcal G_\gamma(q^{-1} x)/\mathcal G_\gamma(x) = x^\gamma\) and \((q z;q,q)_{\infty}/(z;q,q)_{\infty}=1/(z;q)_{\infty}\), we have 
\begin{equation*}\label{eq:r(t)}
  r(t) \coloneqq \frac{\tau\tau_{12}}{\tau_1\tau_2}
  = \frac{(1-t)(1-q^{\alpha_1-\alpha_2} t)}{(1-q^{\alpha_1-\beta_1} t)(1-q^{\alpha_1-\beta_2} t)}.
\end{equation*}
Now we take \(A(z;x)\) as in~\eqref{eq:A(z)-parametrization} by choosing its canonical pair \(\bigl(V(x),X\bigr)\) with
\begin{equation*}\label{eq:(V(x),X)}
  V(x) = \frac{1}{\tau}
    \begin{pmatrix}
      \tau_1 & \tau_2 w(x)\\
      \tau_1\dfrac{1-r(t)}{w(x)} & \tau_2
    \end{pmatrix},
    \quad
    w(x) = \omega_* 
    \frac{
      (q^{\alpha_1-\alpha_2}t;q)_\infty
      (q^{\alpha_1-\alpha_2+1}t;q)_\infty
    }{
      (q^{\alpha_1-\beta_1}t;q)_\infty
      (q^{\alpha_1-\beta_2}t;q)_\infty
    } x_2^{\alpha_1-\alpha_2},
\end{equation*}
for some constant \(\omega_* \in \C^*\), then we have
\begin{equation}\label{eq:A(z;x)}
  A(z;x) = q^{D(\alpha)}X^{-1}
  \left[
    \frac1{r(t)}
    \begin{pmatrix}
      x_1-(1-r(t))x_2 & w(x)(x_2-x_1)\\
      \frac{1-r(t)}{w(x)}(x_1-x_2) & x_2-(1-r(t))x_1
    \end{pmatrix}
    -z\mathbf1_{\mathcal V}
  \right].
\end{equation}
One can verify that \(\bigl(V(x),D(x)\bigr)\) is a flat canonical pair for \(A(z;x)\in E_{\alpha,\beta}(x)\), with local tau function \(\tau(x)\).
Moreover, by setting \(q=1+\varepsilon\), \(x_i(\varepsilon) = - q^{\alpha_i}/(\varepsilon u_i)\), \(\omega_*(\varepsilon) = - \omega \varepsilon^{\alpha_1-\alpha_2+1}\), and writing \(d \coloneqq u_1-u_2\), \(A(z;x)\) formally reduces to a solution \(A(u)\) of \eqref{eq:limiting-isomonodromy-equation} via
\begin{equation*}
  A(z;x(\varepsilon)) = \mathbf1_{\mathcal V} + \varepsilon\bigl(Uz+A(u)\bigr) + O(\varepsilon^2), \quad
  A(u) = \begin{pmatrix}
    \alpha_1 & \omega d^{\alpha_2-\alpha_1}\\
    (\alpha_1\alpha_2-\beta_1\beta_2)\omega^{-1}d^{\alpha_1-\alpha_2} & \alpha_2
  \end{pmatrix}.
\end{equation*}

\section*{Acknowledgments}
The author is grateful to Xiaomeng Xu for comments and discussions, and appreciates the financial and institutional support from the Key Laboratory of Mathematics and Its Applications (Peking University), Ministry of Education.

\bibliographystyle{abbrv}
\bibliography{references}

\Addresses
\end{document}